\documentclass[journal,onecolumn]{IEEEtran}
\usepackage{amsmath, epsfig, cite}
\usepackage{amssymb}
\usepackage{amsfonts}
\usepackage{latexsym}
\usepackage{longtable}
\usepackage{pgf}
\usepackage{tikz}
\usetikzlibrary{arrows, automata, positioning, calc}

\newtheorem{theorem}{Theorem}

\newtheorem{lemma}[theorem]{Lemma}
\newtheorem{conjecture}[theorem]{Conjecture}

\title{Snake-in-the-Box Codes for Rank Modulation under Kendall's $\tau$-Metric}

\author{Yiwei Zhang and Gennian Ge
\thanks{The research of G. Ge was supported by the National Natural Science Foundation of China under Grant No. 61171198 and Grant No. 11431003, the Importation and Development of High-Caliber Talents Project of Beijing Municipal Institutions, and Zhejiang Provincial Natural Science Foundation of China under Grant No. LZ13A010001.}
\thanks{Y. Zhang is with the School of Mathematical Sciences, Zhejiang University,
Hangzhou 310027, China (e-mail: rexzyw@163.com).}
\thanks{G. Ge is with the School of Mathematical Sciences, Capital Normal University,
Beijing 100048, China. He is also with Beijing Center for Mathematics and Information Interdisciplinary Sciences, Beijing 100048, China (e-mail: gnge@zju.edu.cn).}
}

\begin{document}

\maketitle

\begin{abstract}
For a Gray code in the scheme of rank modulation for flash memories, the codewords are permutations and two consecutive codewords are obtained using a push-to-the-top operation. We consider snake-in-the-box codes under Kendall's $\tau$-metric, which is a Gray code capable of detecting one Kendall's $\tau$-error. We answer two open problems posed by Horovitz and Etzion. Firstly, we prove the validity of a construction given by them, resulting in a snake of size $M_{2n+1}=\frac{(2n+1)!}{2}-2n+1$. Secondly, we come up with a different construction aiming at a longer snake of size $M_{2n+1}=\frac{(2n+1)!}{2}-2n+3$. The construction is applied successfully to $S_7$.
\end{abstract}


\section{Introduction}
Flash memory is a non-volatile storage medium both electrically programmable and erasable. It is currently widely used due to its reliability, high storage density and relatively low cost. It incorporates a set of cells maintained at a set of levels of charge to encode information. The chief disadvantage of flash memories is their inherent asymmetry between cell programming (injecting cells with charge) and cell erasing (removing charge from cells). While raising the charge level of a cell is an easy operation, reducing the charge level from a single cell is very difficult. In the current technology, the process of a charge reducing operation requires completely erasing a whole large block to which the cell belongs and then reprogramming, which will limit the lifetime of a flash memory. Therefore, over-programming (increasing charge level on a cell above the desired amount) is a severe problem. For this reason, during a programming cycle in real application, charge is injected over several iterations, gradually approaching the designated level. This process will be time-consuming. Moreover, flash memories meet common errors due to charge leakage and reading disturbance.

In order to overcome these problems, the novel framework of {\it rank modulation} is introduced in \cite{Jiang}. Instead of encoding information with the absolute values of charge levels, data is represented by the relative rankings of the charge levels on a group of cells. That is, if we have $n$ cells and $c_1,c_2,\dots,c_n \in \mathbb{R}$ represent the charge levels, then this group of cells is said to encode the permutation $\sigma\in S_n$ such that $c_{\sigma(1)} > c_{\sigma(2)} > \dots > c_{\sigma(n)}$. In this framework, we save us the trouble to deal with errors caused by injection of extra charge or due to charge leakage which only affect the absolute values of charge levels but do not affect the relative rankings. However, sometimes the errors in the charge levels may be large enough to cause some disturbance in the relative rankings. To detect and/or correct such errors we need an appropriate distance measure. Several metrics on permutations are used for this purpose such as Kendall's $\tau$-metric \cite{Barg}, \cite{Jiang2}, \cite{Mazumdar} and the $l_\infty$-metric \cite{Klove}, \cite{Tamo}. In this paper we will only focus on Kendall's $\tau$-metric.

The Kendall's $\tau$-distance \cite{Kendall} between two permutations $\pi_1$ and $\pi_2$ in $S_n$ is the minimum number of adjacent transpositions required to obtain $\pi_2$ from $\pi_1$, where an adjacent transposition is an exchange of two distinct adjacent elements. For example, the Kendall's $\tau$-distance between $\pi_1=[1,2,3,4]$ and $\pi_2=[2,3,1,4]$ is $2$ as we may do the adjacent transpositions $[1,2,3,4]\rightarrow[2,1,3,4]\rightarrow[2,3,1,4]$. Distance one between two permutations indicates an exchange of two adjacent cells, due to a small change in their charge levels which switches their relative ranking. It is further suggested firstly in \cite{Jiang}, and later in \cite{Gad2}, \cite{Wang}, that the only programming operation allowed is raising the charge level of a cell above all the other cells, which is called a ``push-to-the-top" operation. In this manner, over-programming is no longer an issue.

Gray codes using the ``push-to-the-top" operations under Kendall's $\tau$-metric will be the main objective of this rank modulation scheme. The Gray code is first introduced in \cite{Gray} and an excellent survey on Gray codes is given in \cite{Savage}. If we do not consider any distance restriction among codewords, then Jiang {\it et al.}\cite{Jiang} present Gray codes traversing the entire set of permutations. The usage of Gray codes for rank modulation is also discussed in \cite{Gad},\cite{Gad2} and \cite{Jiang2}. Gray codes for rank modulation which detect a single error under a given metric are known as the snake-in-the-box codes. Snake-in-the-box codes are usually discussed in the context of binary codes in the Hamming scheme (see \cite{Abbott} and references therein).

It is of our desire to construct snake-in-the-box codes as large as possible. Yehezkeally and Schwartz \cite{Yehezkeally} give an inductive construction of a snake-in-the-box code under Kendall's $\tau$-metric of size $M_{2n+1}=(2n+1)(2n-1)M_{2n-1}$ in $S_{2n+1}$, using a code of size $M_{2n-1}$ in $S_{2n-1}$. In \cite{Yehezkeally} they also deal with the problem under the $l_\infty$-metric. Later Horovitz and Etzion \cite{Horovitz} improve the inductive construction to $M_{2n+1}=((2n+1)2n-1)M_{2n-1}$, where the initial code is of size $57$ in $S_5$. They also propose a direct construction aiming at a snake of size $\frac{(2n+1)!}{2}-2n+1$ and it is applied successfully to $S_7$ and $S_9$ via computer search. They conjecture that this framework can work for all odd integers and leave it as an open problem. They also ask the problem if there is a better construction. In this paper, we give a rigorous proof for their construction. Then we also come up with a new construction aiming at a longer snake of size $M_{2n+1}=\frac{(2n+1)!}{2}-2n+3$, which is applied successfully to $S_7$. Thus, we answer the two open problems posed by Horovitz and Etzion.

The rest of the paper is organized as follows. In Section \ref{sec2} we define the basic concepts of snake-in-the-box codes in the rank modulation scheme. In Section \ref{sec3} we restate the construction by Horovitz and Etzion. In Section \ref{sec4} we give a proof verifying the validity of their construction. In Section \ref{sec5} we propose our new construction and give a longer snake-in-the-box code in $S_7$ and we conjecture that it can be applied to $S_{2n+1}$ for any $n\ge3$. We conclude the paper in Section \ref{sec6}.

\section{Preliminaries} \label{sec2}

In this section we follow \cite{Horovitz} and \cite{Yehezkeally} to give some definitions and notations for the snake-in-the-box codes in the rank modulation scheme.

Let $[n]$ denote $\{1,2,\dots,n\}$. Let $\pi=[a_1,a_2,\dots,a_n]$ be a permutation over $[n]$ such that for each $i\in[n]$ we have that $\pi(i)=a_i$. This form is known as the {\it vector notation} for permutations. Another useful notation to describe a permutation is its {\it cyclic notation}, where a permutation is expressed as a product of disjoint cycles corresponding to its orbits. For example, the vector notation $[3,4,5,2,1]$ is equivalent to the cyclic notation $(135)(24)$. All the permutations form the group $S_n$ known as the symmetric group on $[n]$ with $|S_n|=n!$. For $\sigma,\pi \in S_n$, their composition, denoted by $\sigma\pi$, is the permutation for which $\sigma\pi(i)=\sigma(\pi(i))$ for all $i\in[n]$.

Given a set $\mathcal{S}$ and a subset of transformations $T\subset \{f|f:\mathcal{S}\rightarrow\mathcal{S}\}$, a {\it Gray code} over $\mathcal{S}$ of size $M$, using transformations from $T$, is a sequence $C=(c_0,c_1,\dots,c_{M-1})$ of $M$ distinct elements from $\mathcal{S}$, called {\it codewords}, such that for each $j\in[M-1]$ there exists some $t_j\in T$ for which $c_j=t_j(c_{j-1})$. The Gray code is called {\it cyclic} if we further have some $t\in T$ such that $c_0=t(c_{M-1})$. Throughout this paper we only focus on cyclic Gray codes.

In the context of rank modulation for flash memories, $\mathcal{S}=S_n$ and the set of transformations $T$ comprises of {\it push-to-the-top operations}. That is, $T=\{t_2,t_3,\dots,t_{n}\}$ where $t_i$ is defined by
\begin{equation*}
t_i([a_1,\dots,a_{i-1},a_i,a_{i+1},\dots,a_n])=[a_i,a_1,\dots,a_{i-1},a_{i+1},\dots,a_n].
\end{equation*}
and a {\it p-transition} will be an abbreviated notation for a push-to-the-top operation.

A sequence of p-transitions will be called a {\it transitions sequence}. An initial permutation $\pi_0$ and a transitions sequence $t_{x_1},t_{x_2},\dots,t_{x_l}$, $x_i\in\{2,3,\dots,n\}$, $1\le i \le l$ together define a sequence of permutations $\pi_0,\pi_1,\dots,\pi_{l-1},\pi_l$, where $\pi_i=t_{x_i}(\pi_{i-1})$ for each $i,1\le i \le l$. This sequence is a cyclic Gray code if $\pi_l=\pi_0$ and for each $0\le i<j\le l-1$, $\pi_i\neq\pi_j$.

Given a permutation $\pi=[a_1,a_2,\dots,a_n]\in S_n$, an {\it adjacent transposition} is an exchange of two adjacent elements $a_i,a_{i+1}$, for some $1\le i \le n-1$, resulting in the permutation $[a_1,\dots,a_{i-1},a_{i+1},a_i,a_{i+2},\dots,a_n]$. The {\it Kendall's $\tau$-distance} between two permutations $\sigma$ and $\pi$, denoted by $d_\mathcal{K}(\sigma,\pi)$, is the minimum number of adjacent transpositions required to transform one permutation into the other. A {\it snake-in-the-box code} is a Gray code with further restriction that any two permutations in the code have their Kendall's $\tau$-distance at least two. That is, it is capable of detecting one Kendall's $\tau$-error. We will call such a snake-in-the-box code a $\mathcal{K}$-snake. We further denote a $\mathcal{K}$-snake of size $M$ with permutations from $S_n$ as an $(n,M,\mathcal{K})$-snake. A $\mathcal{K}$-snake can be represented by listing either the whole sequence of codewords, or the transitions sequence along with the initial permutation.

In \cite{Yehezkeally} it is proved that a Gray code with permutations from $S_n$ using only p-transitions on odd indices is a $\mathcal{K}$-snake. By starting with an even permutation and using only p-transitions on odd indices we get a sequence of even permutations, i.e., a subset of $A_n$, the alternating group of order $n$. This observation saves us the need to check whether a Gray code is in fact a $\mathcal{K}$-snake, at the cost of restricting the permutations in the $\mathcal{K}$-snake to the set of even permutations. However, the cost is not a severe problem since that the following assertions are also proved in \cite{Yehezkeally}.

$\bullet$ If $C$ is an $(n,M,\mathcal{K})$-snake then $M\le \frac{|S_n|}{2}$;

$\bullet$ If $C$ is an $(n,M,\mathcal{K})$-snake which contains a p-transition on an even index then $M\le \frac{|S_n|}{2}-\frac{1}{n-1}{{\lfloor n/2\rfloor-1} \choose 2}$.

This motivates not to use p-transitions on even indices. Since we merely use p-transitions on odd indices, we will only talk about snake-in-the-box codes in $S_{2n+1}$.

\section{The construction of Horovitz and Etzion} \label{sec3}

In this section we restate a direct construction of Horovitz and Etzion in \cite{Horovitz}, aiming at a $\mathcal{K}$-snake of size $M_{2n+1}=\frac{(2n+1)!}{2}-2n+1$. They conjecture that the construction is valid for all odd integers $2n+1\ge5$ and verify the validity for $S_5$, $S_7$ and $S_9$ via computer search.

Firstly, we make a partition on $A_{2n+1}$ into disjoint classes according to the last two ordered elements in the permutation. That is, a class denoted as $[x,y]$ consists of all the even permutations $\pi=[a_1,a_2,\dots,a_{2n+1}]\in A_{2n+1}$ with $a_{2n}=x$ and $a_{2n+1}=y$. There are totally $2n(2n+1)$ classes and each class contains $\frac{(2n-1)!}{2}$ permutations. We further divide each class into $\frac{(2n-2)!}{2}$ subclasses according to the cyclic order of the first $2n-1$ elements in the permutations. Denote each subclass in a class, say $[x,y]$, by $[\alpha]-[x,y]$ where $\alpha$ is the cyclic order of the first $(2n-1)$ elements. (Note that in the sequel the letters $\alpha,\beta,\gamma\dots$ in a vector notation for a permutation stand for a bunch of numbers, possibly just one number or even empty, and its size and contents can be easily inferred by contexts.) For example, a class [1,2] in $S_7$ consists of all permutations $\pi=[a_1,a_2,\dots,a_7]$ ending with $a_6=1$ and $a_7=2$. And therein a subclass $[3,4,5,6,7]-[1,2]$ consists of the permutations $(3,4,5,6,7,1,2),(7,3,4,5,6,1,2),(6,7,3,4,5,1,2),(5,6,7,3,4,1,2)$ and $(4,5,6,7,3,1,2)$. Obviously such a subclass constitutes a $\mathcal{K}$-snake with the transitions sequence consisting of $(2n-1)$ p-transitions $t_{2n-1}$. From now on we refer to this structure as a {\it necklace}.

The next procedure is to merge some necklaces into a larger $\mathcal{K}$-snake. To do this, we have to follow some rules and the rules are described by the following 3-uniform hypergraph, which is of vital importance to the construction.

Define the 3-uniform hypergraph $H_{2n+1}=(V_{2n+1},E_{2n+1})$ as follows. The vertices correspond to all the classes $[x,y]$ of $S_{2n+1}$. For any distinct $x,y,z\in[2n+1]$, an edge named $\langle x,y,z\rangle$ connects the vertices $[x,y]$, $[y,z]$ and $[z,x]$. A {\it nearly spanning tree} $T_{2n+1}$ on this hypergraph is a tree containing all the vertices except for the vertex [2,1]. For example, we may choose $T_5$ containing the following nine edges: $\langle 1,2,3\rangle,\langle 1,2,4\rangle,\langle1,2,5\rangle,\langle1,5,3\rangle,\langle2,3,5\rangle,\langle1,3,4\rangle,\langle2,4,3\rangle,\langle1,4,5\rangle,\langle2,5,4\rangle$. $T_{2n+1}$ can be recursively constructed from $T_{2n-1}$ by adding the following edges: the edges $\langle x,x+1,2n\rangle$ for each $x$, $2\le x \le 2n-2$, the edges $\langle x,x+1,2n+1\rangle$ for each $x$, $2\le x \le 2n-2$ and then the edges $\langle1,2,2n\rangle,\langle1,2n,2n-1\rangle,\langle1,2n+1,2n-1\rangle,\langle1,2n,2n+1\rangle,\langle2,2n+1,2n\rangle$. The following Figure \ref{tree} which appears in \cite{Horovitz} illustrates how to get $T_7$ from $T_5$. The rectangles and circles represent the edges and vertices in $T_5$ respectively while the dashed rectangles and double circles represent the edges and vertices added to obtain $T_7$.

\begin{figure}[h]
\centering
\includegraphics[height=12cm,width=20cm]{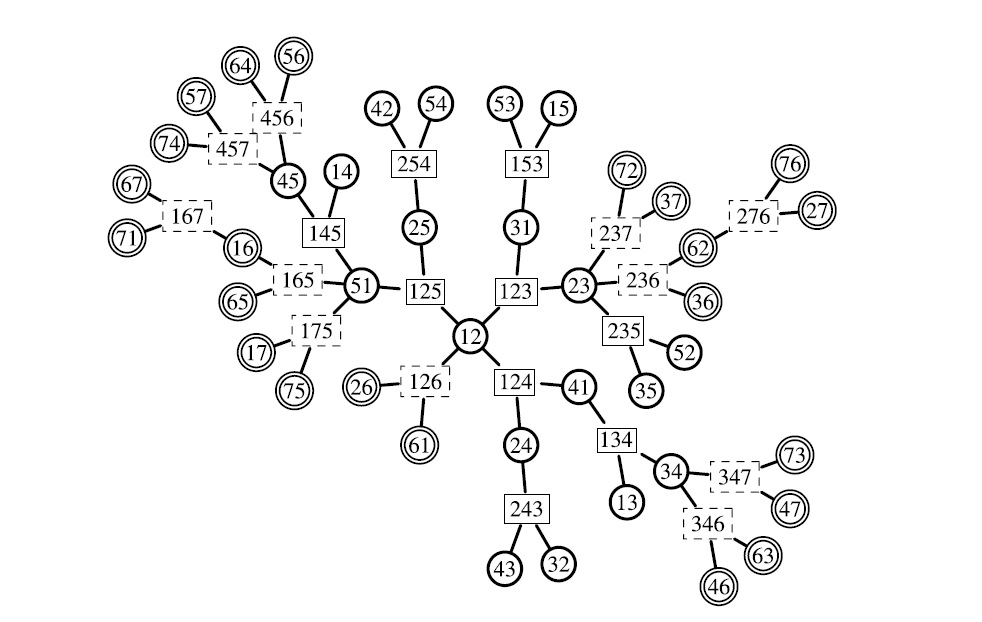}
\caption{Obtaining $T_7$ from $T_5$.} \label{tree}
\end{figure}

After defining the nearly spanning tree $T_{2n+1}$, we now state the rule given by the tree to merge necklaces into a larger $\mathcal{K}$-snake. Start from any necklace $[\alpha]-[1,2]$ in the class $[1,2]$. We choose the edges in $T_{2n+1}$ sequentially (according to the sequence given above). When meeting the edge $\langle x,y,z\rangle$, the already constructed $\mathcal{K}$-snake must contain exactly only one necklace in the union of classes $[x,y]$, $[y,z]$ and $[z,x]$. Without loss of generality we assume an $[x,y]$-necklace belongs to the $\mathcal{K}$-snake. Now we want to merge a $[y,z]$-necklace and a $[z,x]$-necklace into the $\mathcal{K}$-snake. Split the already constructed $\mathcal{K}$-snake at the position right after $[\beta,z,x,y]$ where $\beta$ represents the first $(2n-2)$ elements of the permutation. Such a position surely exists since the existing $[x,y]$-necklace is a cyclic structure on the first $(2n-1)$ positions. We then insert a $[y,z]$-necklace and a $[z,x]$-necklace here as follows. At the splitting point, make a p-transition $t_{2n+1}$ and get $[y,\beta,z,x]$. Then write the whole $[z,x]$-necklace which starts from $[y,\beta,z,x]$ and ends up with $[\beta,y,z,x]$. Another p-transition $t_{2n+1}$ gives $[x,\beta,y,z]$ followed by the whole $[y,z]$-necklace ending up with $[\beta,x,y,z]$. A final p-transition $t_{2n+1}$ will lead us back to $[z,\beta,x,y]$ which is exactly the original permutation right after the splitting point. An example is shown in Figure \ref{57}, giving a $\mathcal{K}$-snake of size 57 in $S_5$. The predefined nearly spanning tree allows us to finally construct a $\mathcal{K}$-snake, containing exactly one necklace in each class $[x,y]$ except for $[2,1]$. From now on we refer to this structure as a {\it chain}. A chain can be constructed as above by choosing any initial necklace $[\alpha]-[1,2]$ and we name this chain as $c[\alpha]$. And it is shown in \cite[Corollary 4]{Horovitz} that the permutations of all the classes except for $[2,1]$ can be partitioned into disjoint chains.

\begin{figure}[h]
\centering
\begin{tikzpicture}[scale=1]
     \tikzstyle{edge} = [draw,thick,->,black]
     \tikzstyle{point}= [fill=black,inner sep=1pt, circle, minimum width=5pt]
	 \node at (-18,0.6)  {$3|5|4$};
     \node at (-18,0.3)  {$4|3|5$};
     \node at (-18,0)  {$5|4|3$};
     \node at (-18,-0.3) {$1|1|1$};
     \node at (-18,-0.6) {$2|2|2$};
     \node at (-17,-1) {$\uparrow$insertion};
     \draw[edge] (-17,0) -- (-16,0);
     \node at (-14,0.6)  {$3|5|4|2|5|4|1|5|4$};
     \node at (-14,0.3)  {$4|3|5|4|2|5|4|1|5$};
     \node at (-14,0)    {$5|4|3|5|4|2|5|4|1$};
     \node at (-14,-0.3) {$1|1|1|3|3|3|2|2|2$};
     \node at (-14,-0.6) {$2|2|2|1|1|1|3|3|3$};
     \node at (-14.08,-1) {$\uparrow$insertion};
     \draw[edge] (-12.2,0) -- (-11.2,0);
     \node at (-8.6,0.6)  {$3|5|2|3|5|1|3|5|4|2|5|4|1|5|4$};
     \node at (-8.6,0.3)  {$4|3|5|2|3|5|1|3|5|4|2|5|4|1|5$};
     \node at (-8.6,0)    {$5|4|3|5|2|3|5|1|3|5|4|2|5|4|1$};
     \node at (-8.6,-0.3) {$1|1|4|4|4|2|2|2|1|3|3|3|2|2|2$};
     \node at (-8.6,-0.6) {$2|2|1|1|1|4|4|4|2|1|1|1|3|3|3$};
     \draw[edge] (-6,0) -- (-5,0);
     \node[point] (v1) at (-4.5,0) {$$};
     \node[point] (v2) at (-4,0) {$$};
     \node[point] (v3) at (-3.5,0) {$$};
\end{tikzpicture}
\begin{tikzpicture}[scale=2]
	 \node at (0,0.4)  {$3|2|1|3|2|5|3|2|4|3|1|5|3|1|2|3|1|4|3|5|2|1|5|2|4|5|2|3|5|1|4|5|1|2|5|1|3|5|4|2|1|4|2|3|4|2|5|4|1|3|4|1|2|4|1|5|4$};
     \node at (0,0.2)  {$4|3|2|1|3|2|5|3|2|4|3|1|5|3|1|2|3|1|4|3|5|2|1|5|2|4|5|2|3|5|1|4|5|1|2|5|1|3|5|4|2|1|4|2|3|4|2|5|4|1|3|4|1|2|4|1|5$};
     \node at (0,0)  {$5|4|3|2|1|3|2|5|3|2|4|3|1|5|3|1|2|3|1|4|3|5|2|1|5|2|4|5|2|3|5|1|4|5|1|2|5|1|3|5|4|2|1|4|2|3|4|2|5|4|1|3|4|1|2|4|1$};
     \node at (0,-0.2) {$1|5|4|4|4|1|1|1|5|5|2|4|4|4|5|5|5|2|2|1|4|3|3|3|1|1|1|4|4|2|3|3|3|4|4|4|2|2|1|3|5|5|5|1|1|1|3|3|2|5|5|5|3|3|3|2|2$};
     \node at (0,-0.4) {$2|1|5|5|5|4|4|4|1|1|5|2|2|2|4|4|4|5|5|2|1|4|4|4|3|3|3|1|1|4|2|2|2|3|3|3|4|4|2|1|3|3|3|5|5|5|1|1|3|2|2|2|5|5|5|3|3$};
\end{tikzpicture}
\caption{Merging necklaces into chains, $M_5=57$.} \label{57}
\end{figure}
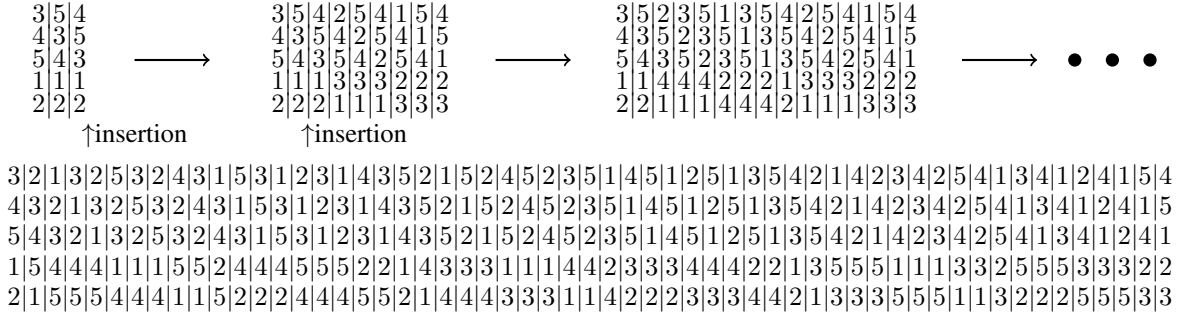

So far we have totally $\frac{(2n-2)!}{2}$ chains using up all the permutations from all classes except for the class $[2,1]$. The next procedure is to apply these unused necklaces in the class [2,1] to merge these chains into a larger $\mathcal{K}$-snake. The following lemma is proved in \cite[Lemma 11]{Horovitz}.

\begin{lemma} \label{merge}
Let $x$ be an integer such that $3\le x \le 2n+1$, let $\alpha$ be a permutation on $[2n+1]\backslash\{x,1,2\}$, and assume that the permutations $[\alpha,1,x,2]$ and $[\alpha,2,1,x]$ are contained in two distinct chains. We can merge these two chains via the necklace $[\alpha,x]-[2,1]$.
\end{lemma}

The merging procedure above is called an {\it $M[x]$-connection} and we call the necklace $[\beta]-[2,1]$ as a {\it linkage} where $\beta$ represents the cyclic order of $(\alpha,x)$. The merging procedure is shown in the following Figure \ref{link}. 

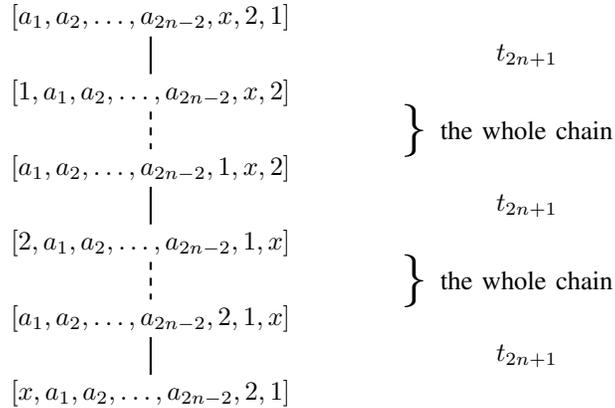
\begin{figure}[h]
\centering
\begin{tikzpicture}[scale=0.5]
     \tikzstyle{edge1} = [draw,thick,-,black]
     \tikzstyle{edge2} = [draw,thick,dashed,-,black]
	 \node at (0,5)  {$[a_1,a_2,\dots,a_{2n-2},x,2,1]$};
     \draw[edge1] (0,4.5) -- (0,3.5);
     \node at (10,4) {$t_{2n+1}$};
     \node at (0,3)  {$[1,a_1,a_2,\dots,a_{2n-2},x,2]$};
     \draw[edge2] (0,2.5) -- (0,1.5);
     \node at (7,2) [scale=2] {$\}$};
     \node at (10,2) {the whole chain};
     \node at (0,1)  {$[a_1,a_2,\dots,a_{2n-2},1,x,2]$};
     \draw[edge1] (0,0.5) -- (0,-0.5);
     \node at (10,0) {$t_{2n+1}$};
     \node at (0,-1) {$[2,a_1,a_2,\dots,a_{2n-2},1,x]$};
     \draw[edge2] (0,-1.5) -- (0,-2.5);
     \node at (7,-2) [scale=2] {$\}$};
     \node at (10,-2) {the whole chain};
     \node at (0,-3) {$[a_1,a_2,\dots,a_{2n-2},2,1,x]$};
     \draw[edge1] (0,-3.5) -- (0,-4.5);
     \node at (10,-4) {$t_{2n+1}$};
     \node at (0,-5) {$[x,a_1,a_2,\dots,a_{2n-2},2,1]$};
\end{tikzpicture}
\caption{An $M[x]$-connection.} \label{link}
\end{figure}

In \cite{Horovitz} the authors mention without proof that if $x\in \{3,4,5\}$ then the permutations $[\alpha,1,x,2]$ and $[\alpha,2,1,x]$ are contained in the same chain, and thus there are no $M[3]$-connections, $M[4]$-connections or $M[5]$-connections. This is actually due to the structure of the nearly spanning tree we choose. We now explain this in detail, together with some other facts concerning $M[x]$-connections for $x>5$.

\begin{theorem} \label{M[X]}
There are no $M[x]$-connections for $x=3,4,5$. For any linkage $[\pi]-[2,1]$ and any $x=2t>5$, $y=2t+1>5$, the $M[x]$ connection via $[\pi]-[2,1]$ connects the chains $[(3x)\pi]-[1,2]$ and $[\sigma\pi]-[1,2]$ while the $M[y]$-connection via $[\pi]-[2,1]$ connects the chains $[(3y)\pi]-[1,2]$ and $[\varsigma\pi]-[1,2]$, where $\sigma$ and $\varsigma$ are permutations on $\{3,4,\dots,2n+1\}$ and using the cyclic notation we have $\sigma=(567\cdots(2t-1)(2t))$ and $\varsigma=(567\cdots(2t-1)(2t+1))$.
\end{theorem}

\begin{proof}
The merging rule suggested by the nearly spanning tree actually indicates that for any edge $\langle x,y,z\rangle$ in $T_{2n+1}$, the necklaces $[\beta,x]-[y,z]$, $[\beta,y]-[z,x]$ and $[\beta,z]-[x,y]$ are merged into the same chain. It is then straight forward to trace back and find the name of the chain to which a certain necklace or a certain permutation belongs.

For example, let $x=3$. We specify the position of the element ``4" and write the permutation $[\alpha,1,3,2]$ as $\pi_1=[\beta,4,\gamma,1,3,2]$. $\pi_1$ belongs to the same necklace as $\pi_2=[\gamma,1,\beta,4,3,2]$. The edge $\langle2,4,3\rangle$ indicates this necklace is in the same chain as the necklace containing $\pi_3=[\gamma,1,\beta,3,2,4]$. $\pi_3$ belongs to the same necklace as $\pi_4=[\beta,3,\gamma,1,2,4]$. Finally the edge $\langle1,2,4\rangle$ indicates we have the necklace containing $[\beta,3,\gamma,4,1,2]$ in this chain. So the permutation $[\alpha,1,3,2]$ is contained in the chain $c[\beta,3,\gamma,4]$.

Similarly, write the permutation $[\alpha,2,1,3]$ as $\sigma_1=[\beta,4,\gamma,2,1,3]$. $\sigma_1$ belongs to the same necklace as $\sigma_2=[\gamma,2,\beta,4,1,3]$. The edge $\langle1,3,4\rangle$ indicates this necklace is in the same chain as the necklace containing $\sigma_3=[\gamma,2,\beta,3,4,1]$. $\sigma_3$ belongs to the same necklace as $\sigma_4=[\beta,3,\gamma,2,4,1]$. Finally the edge $\langle1,2,4\rangle$ indicates we have the necklace containing $[\beta,3,\gamma,4,1,2]$ in this chain. So the permutation $[\alpha,2,1,3]$ is contained in the chain $c[\beta,3,\gamma,4]$. Summing up the above we conclude that the permutations $[\alpha,1,3,2]$ and $[\alpha,2,1,3]$ are in the same chain. For $x=4,5$ we have a similar procedure. So there are no $M[x]$-connections for $x=3,4,5$.

The remaining statement can be analyzed similarly and we only do as an example for $x=6$ with any linkage $[\pi]-[2,1]=[\alpha,6]-[2,1]$. Specify the position of ``3" and write $[\alpha,1,6,2]$ as $\pi_1=[\beta,3,\gamma,1,6,2]$. Then we can find in the same chain the following permutations one by one: $[\gamma,1,\beta,3,6,2]$, $[\gamma,1,\beta,6,2,3]$, $[\beta,6,\gamma,1,2,3]$, $[\beta,6,\gamma,3,1,2]$. Since $[\pi]=[\alpha,6]=[\beta,3,\gamma,6]$ so we find the name of the chain to be $[(36)\pi]-[1,2]$. Specify the position of ``5" and write $[\alpha,2,1,6]$ as $\sigma_1=[\beta',5,\gamma',2,1,6]$ and we can find in the same chain the following permutations one by one: $[\gamma',2,\beta',5,1,6]$, $[\gamma',2,\beta',6,5,1]$, $[\beta',6,\gamma',2,5,1]$, $[\beta',6,\gamma',5,1,2]$. Since $[\pi]=[\alpha,6]=[\beta',5,\gamma',6]$ so we find the name of the chain to be $[(56)\pi]-[1,2]$.

The remaining proof for other values of $x$ is proved in a similar but rather tedious way and thus we omit it.
\end{proof}

Define a graph $\mathcal{G}_{2n+1}=(\mathcal{V}_{2n+1},\mathcal{E}_{2n+1})$ where the vertices represent the set of chains. Two chains are connected by an edge if and only if they can be merged as Lemma \ref{merge}. Each edge has a {\it sign} $M[x]$ (indicating the merging is an $M[x]$-connection) and a {\it label} $[\alpha,x]-[2,1]$ (indicating the name of the linkage). The problem of merging all chains into a large snake reduces to finding a spanning tree $\mathcal{T}_{2n+1}$ in $\mathcal{G}_{2n+1}$ such that all edges have distinct labels. We require distinct labels since we want to use as many $[2,1]$-necklaces as possible (all except one). Once the spanning tree is found then we are able to merge all the chains and all except one $[2,1]$-necklaces into a $\mathcal{K}$-snake of size $M_{2n+1}=\frac{(2n+1)!}{2}-2n+1$. Horovitz and Etzion \cite{Horovitz} conjecture that the desired spanning tree always exists and verify for $S_7$ and $S_9$ via computer search. We proceed in the next section to give a construction of the spanning tree and thus complete their framework.

It should be remarked that the $\mathcal{K}$-snake constructed this way has an interesting property that its transitions sequence only consists of p-transitions $t_{2n-1}$ and $t_{2n+1}$.

\section{Existence of the spanning tree with distinct labels}  \label{sec4}

We first look into the case $S_7$ as an illustrative example. $\mathcal{G}_7$ consists of $12$ vertices corresponding to the $12$ chains:
\begin{align*}
&c_1=[4,5,6,7,3]-[1,2],     &c_2=[4,6,7,5,3]-[1,2], \\
&c_3=[4,7,5,6,3]-[1,2],     &c_4=[4,7,6,3,5]-[1,2], \\
&c_5=[4,7,3,5,6]-[1,2],     &c_6=[4,3,5,7,6]-[1,2], \\
&c_7=[4,5,7,3,6]-[1,2],     &c_8=[4,3,6,5,7]-[1,2], \\
&c_9=[4,5,3,6,7]-[1,2],     &c_{10}=[4,6,5,3,7]-[1,2], \\
&c_{11}=[4,6,3,7,5]-[1,2],      &c_{12}=[4,3,7,6,5]-[1,2].
\end{align*}
The $12$ linkages ($[2,1]$-necklaces) are:
\begin{align*}
&\eta_1=[4,5,7,6,3]-[2,1],     &\eta_2=[4,6,5,7,3]-[2,1], \\
&\eta_3=[4,7,6,5,3]-[2,1],     &\eta_4=[4,6,7,3,5]-[2,1], \\
&\eta_5=[4,3,5,6,7]-[2,1],     &\eta_6=[4,6,3,5,7]-[2,1], \\
&\eta_7=[4,7,5,3,6]-[2,1],     &\eta_8=[4,7,3,6,5]-[2,1], \\
&\eta_9=[4,3,6,7,5]-[2,1],     &\eta_{10}=[4,5,6,3,7]-[2,1], \\
&\eta_{11}=[4,3,7,5,6]-[2,1],     &\eta_{12}=[4,5,3,7,6]-[2,1].
\end{align*}

As Theorem \ref{M[X]} indicates, $\mathcal{G}_7$ will only contain edges with signs $M[6]$ and $M[7]$. By an $M[6]$-connection, a linkage $[\alpha]-[2,1]$ will connect the chains $[(36)\alpha]-[1,2]$ and $[(56)\alpha]-[1,2]$. Similarly by an $M[7]$-connection, a linkage $[\alpha]-[2,1]$ will connect the chains $[(37)\alpha]-[1,2]$ and $[(57)\alpha]-[1,2]$. Note that we present the chains and linkages above in the exact same order as in \cite{Horovitz}. The difference is that while they present each chain $[\alpha]-[1,2]$ or linkage $[\alpha]-[2,1]$ with $\alpha$ starting from ``3", we instead start from ``4" since it benefits the upcoming analysis.

Now we rename the chains and linkages according to the positions of ``6" and ``7". Suppose ``6" is on the $i$-th position and ``7" is on the $j$-th position. Note that we also have fixed ``4" on the first position. Then a unique chain/linkage will be determined since there will be only one choice to place ``3" and ``7" to get an even permutation. Denote this chain/linkage by $C_{i,j}$ and $L_{i,j}$ respectively for $2\le i,~j\le5$ and $i\neq j$. Then, by an $M[6]$-connection, a linkage $L_{i,j}$ will connect the chains $C_{k,j}$ and $C_{l,j}$ where $k$ and $l$ are the two elements in $\{2,3,4,5\}\backslash \{i,j\}$. Similarly, by an $M[7]$-connection, a linkage $L_{i,j}$ will connect the chains $C_{i,k}$ and $C_{i,l}$ where $k$ and $l$ are the two elements in $\{2,3,4,5\}\backslash \{i,j\}$. Figure \ref{G7} shows the structure of $\mathcal{G}_7$. The next goal is to find a spanning tree $\mathcal{T}_7$ with distinct labels. To do this we first strengthen to find a cycle $\mathcal{C}_7$ with distinct labels connecting all the vertices, and then we delete any edge in the cycle to get a spanning tree as desired. This technique is key to the analysis later. The cycle can be chosen as: for any linkage $(i,j)$ with $j\equiv i-1 \pmod{5}$ we choose the edge corresponding to its $M[6]$-connection and for the other linkages we choose their $M[7]$-connections. The resulting cycle is shown in Figure \ref{G7}. Deleting any edge in this cycle, we get a spanning tree indicating the method to merge all the chains and all but one linkages into a whole $\mathcal{K}$-snake of size $M_7=2515$. Note that the only absent five permutations are those permutations in the linkage corresponding to the edge deleted.

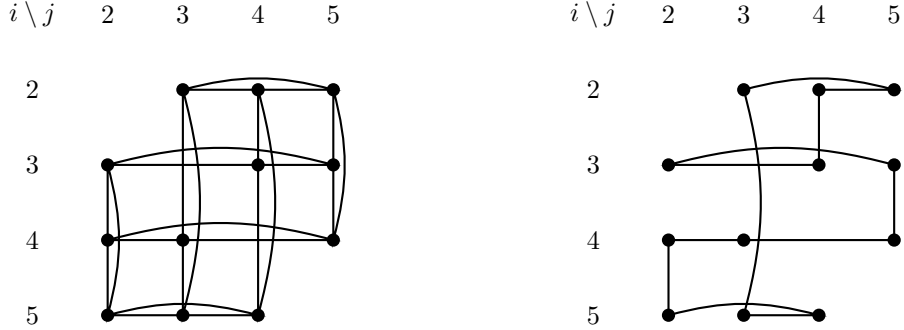
\begin{figure}[h]
\centering
\begin{tikzpicture}[scale=1]
     \tikzstyle{edge} = [draw,thick,-,black]
     \tikzstyle{point}= [fill=black,inner sep=1pt, circle, minimum width=5pt]
	 \node at (-2,2) {$i\setminus j$};
     \node at (-1,2)  {$2$};
     \node at (0,2)  {$3$};
     \node at (1,2)  {$4$};
     \node at (2,2)  {$5$};
     \node at (-2,1)  {$2$};
     \node at (-2,0)  {$3$};
     \node at (-2,-1)  {$4$};
     \node at (-2,-2)  {$5$};
     \node[point] (v1) at (0,1) {$$};
     \node[point] (v2) at (1,1) {$$};
     \node[point] (v3) at (2,1) {$$};
     \node[point] (v4) at (-1,0) {$$};
     \node[point] (v5) at (1,0) {$$};
     \node[point] (v6) at (2,0) {$$};
     \node[point] (v7) at (-1,-1) {$$};
     \node[point] (v8) at (0,-1) {$$};
     \node[point] (v9) at (2,-1) {$$};
     \node[point] (v10) at (-1,-2) {$$};
     \node[point] (v11) at (0,-2) {$$};
     \node[point] (v12) at (1,-2) {$$};
     \draw[edge] (0,1)-- (1,1) -- (2,1) to [bend right=15] (0,1);
     \draw[edge] (-1,0)-- (1,0) -- (2,0) to [bend right=15] (-1,0);
     \draw[edge] (-1,-1)-- (0,-1) -- (2,-1) to [bend right=15] (-1,-1);
     \draw[edge] (-1,-2)-- (0,-2) -- (1,-2) to [bend right=15] (-1,-2);
     \draw[edge] (-1,0)-- (-1,-1) -- (-1,-2) to [bend right=15] (-1,0);
     \draw[edge] (0,1)-- (0,-1) -- (0,-2) to [bend right=15] (0,1);
     \draw[edge] (1,1)-- (1,0) -- (1,-2) to [bend right=15] (1,1);
     \draw[edge] (2,1)-- (2,0) -- (2,-1) to [bend right=15] (2,1);
\end{tikzpicture}
\begin{tikzpicture}[scale=1]
     \tikzstyle{edge} = [draw,thick,-,black]
     \tikzstyle{point}= [fill=black,inner sep=1pt, circle, minimum width=5pt]
     \node at (-5,0) {$$};
	 \node at (-2,2) {$i\setminus j$};
     \node at (-1,2)  {$2$};
     \node at (0,2)  {$3$};
     \node at (1,2)  {$4$};
     \node at (2,2)  {$5$};
     \node at (-2,1)  {$2$};
     \node at (-2,0)  {$3$};
     \node at (-2,-1)  {$4$};
     \node at (-2,-2)  {$5$};
     \node[point] (v1) at (0,1) {$$};
     \node[point] (v2) at (1,1) {$$};
     \node[point] (v3) at (2,1) {$$};
     \node[point] (v4) at (-1,0) {$$};
     \node[point] (v5) at (1,0) {$$};
     \node[point] (v6) at (2,0) {$$};
     \node[point] (v7) at (-1,-1) {$$};
     \node[point] (v8) at (0,-1) {$$};
     \node[point] (v9) at (2,-1) {$$};
     \node[point] (v10) at (-1,-2) {$$};
     \node[point] (v11) at (0,-2) {$$};
     \node[point] (v12) at (1,-2) {$$};
     \draw[edge] (1,1) -- (2,1) to [bend right=15] (0,1);
     \draw[edge] (2,0) to [bend right=15] (-1,0) -- (1,0);
     \draw[edge] (-1,-1) -- (0,-1) -- (2,-1);
     \draw[edge] (0,-2) -- (1,-2) to [bend right=15] (-1,-2);
     \draw[edge] (-1,-2)-- (-1,-1);
     \draw[edge] (2,-1) -- (2,0);
     \draw[edge] (0,-2) to [bend right=15] (0,1);
     \draw[edge] (1,0)-- (1,1);
\end{tikzpicture}
\caption{$\mathcal{G}_7$ and $\mathcal{C}_7$.} \label{G7}
\end{figure}

After this initial case, the construction of $\mathcal{T}_{2n+1}$ now follows in an inductive way. The induction is due to the following lemma proved in \cite[Lemma 16]{Horovitz}.
\begin{lemma} \label{induction}
For each $n\ge4$, $\mathcal{G}_{2n+1}$ consists of $(2n-3)(2n-2)$ disjoint copies of isomorphic graphs to $\mathcal{G}_{2n-1}$, called components. The edges between the vertices of two distinct components are signed only with $M[2n]$ and $M[2n+1]$.
\end{lemma}

We look deeply into the structure of $\mathcal{G}_{2n+1}$. Let $C_{i,j}$ and $L_{i,j}$ denote respectively the set of all chains and linkages with $(2n)$ on the $i$-th position and $(2n+1)$ on the $j$-th position. As Theorem \ref{link} indicates, $C_{i,j}$ is exactly the so-called component in Lemma \ref{induction} above. The edges corresponding to all linkages in $L_{i,j}$ and all $M[x]$-connections except $x=2n,~2n+1$ are exactly all the edges within $C_{i,j}$. Now, define a graph $\hat{\mathcal{G}}_{2n+1}=(\hat{\mathcal{V}}_{2n+1},\hat{\mathcal{E}}_{2n+1})$ where the vertices correspond to the set $\{C_{ij}:2\le i,j\le 2n-1, i\neq j\}$. For each pair of chains $c_1\in C_{i,j}$ and $c_2\in C_{i',j'}$ such that $c_1$ and $c_2$ are connected in $\mathcal{G}$, draw an edge between $C_{i,j}$ and $C_{i',j'}$ with the same sign and label as the edge connecting $c_1$ and $c_2$ in $\mathcal{G}$. There will be only two signs $M[2n]$ and $M[2n+1]$.

\begin{theorem}
There exists a cycle $\hat{\mathcal{C}}_{2n+1}$ connecting all vertices in $\hat{\mathcal{V}}_{2n+1}$, with the labels coming from distinct $L_{i,j}$.
\end{theorem}

\begin{proof}
For each $L_{i,j}$ with $j\equiv i-1 \pmod{2n-1}$, we choose a linkage in $L_{i,j}$ with ``3" on the $(i-2)$-th position and ``$2n-1$" on the $(i-3)$-th position. Then its $M[2n]$-connection will connect $C_{i-2,j}$ and $C_{i-3,j}$, i.e. connect $C_{i-2,i-1}$ and $C_{i-3,i-1}$.
For each $L_{i,j}$ with $j\equiv i-2 \pmod{2n-1}$, we choose a linkage in $L_{i,j}$ with ``3" on the $(i-1)$-th position and ``$2n-1$" on the $(i+1)$-th position. Then its $M[2n+1]$-connection will connect $C_{i,i-1}$ and $C_{i,i+1}$.
For the other linkages $L_{i,j}$, we choose a linkage in $L_{i,j}$ with ``3" on the $(j+1)$-th position and ``$2n-1$" on the $(j+2)$-th position. Then its $M[2n+1]$-connection will connect $C_{i,j+1}$ and $C_{i,j+2}$.
It is a little tedious but straight forward to check that the edges above constitute the cycle as desired.
\end{proof}

As an illustrative example, the cycle in $\hat{\mathcal{G}}_{9}$ is given in Figure \ref{9}.

\begin{figure}[h]
\centering
\begin{tikzpicture}[scale=1]
     \tikzstyle{edge} = [draw,thick,-,black]
     \tikzstyle{point}= [fill=black,inner sep=1pt, circle, minimum width=5pt]
	 \node at (-3,3) {$i\setminus j$};

     \node at (-2,3)  {$2$};
     \node at (-1,3)  {$3$};
     \node at (0,3)  {$4$};
     \node at (1,3)  {$5$};
     \node at (2,3)  {$6$};
     \node at (3,3)  {$7$};

     \node at (-3,2)  {$2$};
     \node at (-3,1)  {$3$};
     \node at (-3,0)  {$4$};
     \node at (-3,-1)  {$5$};
     \node at (-3,-2)  {$6$};
     \node at (-3,-3)  {$7$};

     \node[point] (v1) at (-1,2) {$$};
     \node[point] (v2) at (0,2) {$$};
     \node[point] (v3) at (1,2) {$$};
     \node[point] (v4) at (2,2) {$$};
     \node[point] (v5) at (3,2) {$$};
     \node[point] (v6) at (-2,1) {$$};
     \node[point] (v7) at (0,1) {$$};
     \node[point] (v8) at (1,1) {$$};
     \node[point] (v9) at (2,1) {$$};
     \node[point] (v10) at (3,1) {$$};
     \node[point] (v11) at (-2,0) {$$};
     \node[point] (v12) at (-1,0) {$$};
     \node[point] (v13) at (1,0) {$$};
     \node[point] (v14) at (2,0) {$$};
     \node[point] (v15) at (3,0) {$$};
     \node[point] (v16) at (-2,-1) {$$};
     \node[point] (v17) at (-1,-1) {$$};
     \node[point] (v18) at (0,-1) {$$};
     \node[point] (v19) at (2,-1) {$$};
     \node[point] (v20) at (3,-1) {$$};
     \node[point] (v21) at (-2,-2) {$$};
     \node[point] (v22) at (-1,-2) {$$};
     \node[point] (v23) at (0,-2) {$$};
     \node[point] (v24) at (1,-2) {$$};
     \node[point] (v25) at (3,-2) {$$};
     \node[point] (v26) at (-2,-3) {$$};
     \node[point] (v27) at (-1,-3) {$$};
     \node[point] (v28) at (0,-3) {$$};
     \node[point] (v29) at (1,-3) {$$};
     \node[point] (v30) at (2,-3) {$$};

     \draw[edge] (0,2)-- (1,2) -- (2,2) -- (3,2) to [bend right=15] (-1,2);
     \draw[edge] (1,1)-- (2,1) -- (3,1) to [bend right=15] (-2,1) -- (0,1);
     \draw[edge] (2,0)-- (3,0) to [bend right=15] (-2,0) -- (-1,0) -- (1,0);
     \draw[edge] (3,-1) to [bend right=15] (-2,-1) -- (-1,-1) -- (0,-1) -- (2,-1);
     \draw[edge] (-2,-2)-- (-1,-2) -- (0,-2) -- (1,-2) -- (3,-2);
     \draw[edge] (-1,-3)-- (0,-3) -- (1,-3) --(2,-3) to [bend right=15] (-2,-3);
     \draw[edge] (0,2) -- (0,1);
     \draw[edge] (1,1) -- (1,0);
     \draw[edge] (2,0) -- (2,-1);
     \draw[edge] (3,-1) -- (3,-2);
     \draw[edge] (-2,-2) -- (-2,-3);
     \draw[edge] (-1,-3) to [bend right=15] (-1,2);

\end{tikzpicture}
\caption{A cycle in $\hat{\mathcal{G}}_{9}$.} \label{9}
\end{figure}

Now the inductive procedure goes as follows. Delete any edge in the cycle $\hat{\mathcal{C}}_{2n+1}$ constructed in $\hat{\mathcal{G}}_{2n+1}$ to get its spanning tree with their labels coming from distinct $L_{i,j}$. Then at most one linkage in $L_{i,j}$ has been occupied in $\hat{\mathcal{C}}_{2n+1}$. $C_{i,j}$ is locally connected by a cycle with distinct labels corresponding to the set of linkages $L_{i,j}$. Deleting the edge corresponding to the occupied linkage, we still have a spanning tree connecting all the chains in $C_{i,j}$. Thus we find a spanning tree with distinct labels for the whole graph $\mathcal{G}_{2n+1}$.

\section{A further improvement on the size of a $\mathcal{K}$-snake}  \label{sec5}

In this section we construct a longer $\mathcal{K}$-snake in $S_7$ of size $M_7=2517$, increasing the construction of Horovitz and Etzion with $M_7=2515$ by $2$.

The basic preparations are exactly the same as the construction above. We first get the $12$ chains which together use up all the permutations except those in the class $[2,1]$. The unused permutations now are those $12$ $[2,1]$-necklaces each of size $5$. Horovitz and Etzion use them as linkages to merge the chains and thus the absence of one of these necklaces is inevitable. How about constructing a $\mathcal{K}$-snake using only the permutations in the class $[2,1]$ first? This is equivalent to constructing a $\mathcal{K}$-snake in $S_5$ and we already have such a $\mathcal{K}$-snake of size 57 in Figure \ref{57}. Now we take some one-to-one map $f:\{1,2,3,4,5\}\rightarrow\{3,4,5,6,7\}$ and add the tails $(2,1)$ to turn the $\mathcal{K}$-snake in $S_5$ into a $\mathcal{K}$-snake in $S_7$. The choice of $f$ should guarantee that the induced $\mathcal{K}$-snake in $S_7$ consists of even permutations.

The next procedure is to insert the $12$ chains into this $\mathcal{K}$-snake. As Lemma \ref{merge} indicates, if the $\mathcal{K}$-snake has two consecutive permutations $[\alpha,x,2,1]$ and $[x,\alpha,2,1]$, $x=6,~7$, then we may insert the two chains containing $[1,\alpha,x,2]$ and $[2,\alpha,1,x]$ respectively. Now if we can find a matching in $\mathcal{G}_7$ whose six edges all correspond to applicable insertions, then we end up with the $\mathcal{K}$-snake of size $2517$ as desired. While there are many matchings in $\mathcal{G}_7$, whether the six edges in a matching all correspond to applicable insertions or not needs to be checked, since the transitions sequence of the $\mathcal{K}$-snake contains lots of p-transitions $t_3$. Ambiguously speaking, the more p-transitions $t_5$, the better. Fortunately, we may do some ``sewing and mending" to the $\mathcal{K}$-snake, due to the fact that $t^{-1}_3 t_5 t^{-1}_3(\pi)=t^{-1}_5 t_3 t^{-1}_5(\pi)$ for every $\pi\in S_7$. We may cut off the segment from $t_3(\pi)$ to $t_3^{-1}t_5(\pi)$, sew $\pi$ and $t_5(\pi)$ together, and then insert the segment at the position between $t^{-1}_5 t_3(\pi)$ and $t_3 t^{-1}_5 t_3(\pi)$ as long as $t^{-1}_5 t_3(\pi)$ and $t_3 t^{-1}_5 t_3(\pi)$ are not within the segment cut off. This modification brings in more p-transitions $t_5$ into the transitions sequence of the $\mathcal{K}$-snake without deleting any existing $t_5$. Now we may insert the $12$ chains in pairs as in Figure \ref{sew}.

\begin{figure}[h]
\centering
\begin{tikzpicture}[scale=2]
     \tikzstyle{edge} = [draw,thick,-,black]
     \tikzstyle{point}= [fill=black,inner sep=1pt, circle, minimum width=5pt]

	 \node at (0,2.4)  {$3|2|1|3|2|5|3|2|4|3|1|5|3|1|2|3|1|4|3|5|2|1|5|2|4|5|2|3|5|1|4|5|1|2|5|1|3|5|4|2|1|4|2|3|4|2|5|4|1|3|4|1|2|4|1|5|4$};
     \node at (0,2.2)  {$4|3|2|1|3|2|5|3|2|4|3|1|5|3|1|2|3|1|4|3|5|2|1|5|2|4|5|2|3|5|1|4|5|1|2|5|1|3|5|4|2|1|4|2|3|4|2|5|4|1|3|4|1|2|4|1|5$};
     \node at (0,2)    {$5|4|3|2|1|3|2|5|3|2|4|3|1|5|3|1|2|3|1|4|3|5|2|1|5|2|4|5|2|3|5|1|4|5|1|2|5|1|3|5|4|2|1|4|2|3|4|2|5|4|1|3|4|1|2|4|1$};
     \node at (0,1.8)  {$1|5|4|4|4|1|1|1|5|5|2|4|4|4|5|5|5|2|2|1|4|3|3|3|1|1|1|4|4|2|3|3|3|4|4|4|2|2|1|3|5|5|5|1|1|1|3|3|2|5|5|5|3|3|3|2|2$};
     \node at (0,1.6)  {$2|1|5|5|5|4|4|4|1|1|5|2|2|2|4|4|4|5|5|2|1|4|4|4|3|3|3|1|1|4|2|2|2|3|3|3|4|4|2|1|3|3|3|5|5|5|1|1|3|2|2|2|5|5|5|3|3$};
     \node at (0,1.3) {$\Downarrow$ The map $f$: $f(1)=5$, $f(2)=6$, $f(3)=3$, $f(4)=7$, $f(5)=4$, then add the tails $\Downarrow$};
     \node at (0,1)    {$3|6|5|3|6|4|3|6|7|3|5|4|3|5|6|3|5|7|3|4|6|5|4|6|7|4|6|3|4|5|7|4|5|6|4|5|3|4|7|6|5|7|6|3|7|6|4|7|5|3|7|5|6|7|5|4|7$};
     \node at (0,0.8)  {$7|3|6|5|3|6|4|3|6|7|3|5|4|3|5|6|3|5|7|3|4|6|5|4|6|7|4|6|3|4|5|7|4|5|6|4|5|3|4|7|6|5|7|6|3|7|6|4|7|5|3|7|5|6|7|5|4$};
     \node at (0,0.6)  {$4|7|3|6|5|3|6|4|3|6|7|3|5|4|3|5|6|3|5|7|3|4|6|5|4|6|7|4|6|3|4|5|7|4|5|6|4|5|3|4|7|6|5|7|6|3|7|6|4|7|5|3|7|5|6|7|5$};
     \node at (0,0.4)  {$5|4|7|7|7|5|5|5|4|4|6|7|7|7|4|4|4|6|6|5|7|3|3|3|5|5|5|7|7|6|3|3|3|7|7|7|6|6|5|3|4|4|4|5|5|5|3|3|6|4|4|4|3|3|3|6|6$};
     \node at (0,0.2)  {$6|5|4|4|4|7|7|7|5|5|4|6|6|6|7|7|7|4|4|6|5|7|7|7|3|3|3|5|5|7|6|6|6|3|3|3|7|7|6|5|3|3|3|4|4|4|5|5|3|6|6|6|4|4|4|3|3$};
     \node at (0,0)    {$2|2|2|2|2|2|2|2|2|2|2|2|2|2|2|2|2|2|2|2|2|2|2|2|2|2|2|2|2|2|2|2|2|2|2|2|2|2|2|2|2|2|2|2|2|2|2|2|2|2|2|2|2|2|2|2|2$};
     \node at (0,-0.2) {$1|1|1|1|1|1|1|1|1|1|1|1|1|1|1|1|1|1|1|1|1|1|1|1|1|1|1|1|1|1|1|1|1|1|1|1|1|1|1|1|1|1|1|1|1|1|1|1|1|1|1|1|1|1|1|1|1$};
     \node at (2.05,-0.4) {$|\star---~cut~---\star|$};
     \node at (-2.95,-0.4) {$\uparrow$ insert here};
     \node at (0,-0.6) {$\Downarrow$};

     \node at (0,-0.8)    {$3|6|5|3|4|7|6|5|7|6|3|7|6|4|7|5|3|6|4|3|6|7|3|5|4|3|5|6|3|5|7|3|4|6|5|4|6|7|4|6|3|4|5|7|4|5|6|4|5|3|7|5|6|7|5|4|7$};
     \node at (0,-1)  {$7|3|6|5|3|4|7|6|5|7|6|3|7|6|4|7|5|3|6|4|3|6|7|3|5|4|3|5|6|3|5|7|3|4|6|5|4|6|7|4|6|3|4|5|7|4|5|6|4|5|3|7|5|6|7|5|4$};
     \node at (0,-1.2)  {$4|7|3|6|5|3|4|7|6|5|7|6|3|7|6|4|7|5|3|6|4|3|6|7|3|5|4|3|5|6|3|5|7|3|4|6|5|4|6|7|4|6|3|4|5|7|4|5|6|4|5|3|7|5|6|7|5$};
     \node at (0,-1.4)  {$5|4|7|7|6|5|3|4|4|4|5|5|5|3|3|6|4|7|5|5|5|4|4|6|7|7|7|4|4|4|6|6|5|7|3|3|3|5|5|5|7|7|6|3|3|3|7|7|7|6|4|4|3|3|3|6|6$};
     \node at (0,-1.6)  {$6|5|4|4|7|6|5|3|3|3|4|4|4|5|5|3|6|4|7|7|7|5|5|4|6|6|6|7|7|7|4|4|6|5|7|7|7|3|3|3|5|5|7|6|6|6|3|3|3|7|6|6|4|4|4|3|3$};
     \node at (0,-1.8)    {$2|2|2|2|2|2|2|2|2|2|2|2|2|2|2|2|2|2|2|2|2|2|2|2|2|2|2|2|2|2|2|2|2|2|2|2|2|2|2|2|2|2|2|2|2|2|2|2|2|2|2|2|2|2|2|2|2$};
     \node at (0,-2) {$1|1|1|1|1|1|1|1|1|1|1|1|1|1|1|1|1|1|1|1|1|1|1|1|1|1|1|1|1|1|1|1|1|1|1|1|1|1|1|1|1|1|1|1|1|1|1|1|1|1|1|1|1|1|1|1|1$};
     \node at (-1.02,-2.2) {$\uparrow$};
     \node at (-1.02,-2.3) {$c_3,c_9$};
     \node at (0.61,-2.2) {$\uparrow$};
     \node at (0.61,-2.3) {$c_2,c_{12}$};
     \node at (1.17,-2.2) {$\uparrow$};
     \node at (1.17,-2.3) {$c_5,c_7$};
     \node at (1.98,-2.2) {$\uparrow$};
     \node at (1.98,-2.3) {$c_4,c_8$};
     \node at (2.95,-2.4) {$\uparrow$};
     \node at (2.95,-2.5) {$c_{10},c_{11}$};
     \node at (3.2,-2.2) {$\uparrow$};
     \node at (3.2,-2.3) {$c_1,c_6$};

\end{tikzpicture}
\caption{Constructing a $\mathcal{K}$-snake of size 2517 in $S_7$.} \label{sew}
\end{figure}

We conjecture that this framework is feasible for all odd integers. Its validity strongly depends on the structure of the $\mathcal{K}$-snakes constructed in the framework of Horovitz and Etzion. We have remarked that a $\mathcal{K}$-snake in $S_{2n-1}$ constructed by Horovitz and Etzion has the property that its transitions sequence only consists of $t_{2n-1}$ and $t_{2n-3}$. Starting from such a $\mathcal{K}$-snake with a properly chosen map $f:\{1,2,\dots,2n-1\}\rightarrow\{3,4,\dots,2n+1\}$ and then adding the tails $(2,1)$, we get a $\mathcal{K}$-snake whose transitions sequence only consists of $t_{2n-1}$ and $t_{2n-3}$. Similarly as above, we may do some ``sewing and mending" to the $\mathcal{K}$-snake, due to the fact that $t^{-1}_{2n-3} t_{2n-1} t^{-1}_{2n-3}(\pi)=t^{-1}_{2n-1} t_{2n-3} t^{-1}_{2n-1}(\pi)$ for every $\pi\in S_{2n+1}$. We may cut off the segment from $t_{2n-3}(\pi)$ to $t_{2n-3}^{-1}t_{2n-1}(\pi)$, sew $\pi$ and $t_{2n-1}(\pi)$ together, and then insert the segment at the position between $t^{-1}_{2n-1} t_{2n-3}(\pi)$ and $t_{2n-3} t^{-1}_{2n-1} t_{2n-3}(\pi)$ as long as $t^{-1}_{2n-1} t_{2n-3}(\pi)$ and $t_{2n-3} t^{-1}_{2n-1} t_{2n-3}(\pi)$ are not within the segment cut off. This modification brings in more p-transitions $t_{2n-1}$ into the transitions sequence of the $\mathcal{K}$-snake without deleting any existing $t_{2n-1}$. The position between two consecutive codewords $(\alpha,x,2,1)$ and $(x,\alpha,2,1)$ for some $x>5$ will work as a choice of inserting the two chains containing $[1,\alpha,x,2]$ and $[2,\alpha,1,x]$ respectively. Besides, $\mathcal{G}_{2n+1}$ has a lot of matchings so it is very possible to find a matching whose edges all correspond to applicable insertions. All these optimistic evidences indicate the validity of this framework. Yet a strict mathematical proof still requires further analysis.

Summing up the above, we have the following conjecture:

\begin{conjecture} \label{conj}
There exists a $(2n+1,M_{2n+1},\mathcal{K})$-snake with $M_{2n+1}=\frac{(2n+1)!}{2}-2n+3$ for every $n\ge 3$.
\end{conjecture}

If we do the same procedure as above from an initial snake in our construction (or possibly some other snakes with the same size), rather than a Horovitz-Etzion snake, there might be a slim chance of doing better! However, the transitions sequence of our snake does not have many p-transitions $t_{2n+1}$, and also lacks applicable ``sewing and mending" modifications. So compared with Conjecture \ref{conj}, the following conjecture is a little pessimistic.

\begin{conjecture}
There exists a $(2n+1,M_{2n+1},\mathcal{K})$-snake with $M_{2n+1}>\frac{(2n+1)!}{2}-2n+3$ or even $M_{2n+1}=\frac{(2n+1)!}{2}-3$ for every $n\ge 3$.
\end{conjecture}

A final remark is that ``greed is part of human nature". The possibility of $M_{2n+1}=\frac{(2n+1)!}{2}$, however impossible, is not yet denied.

\section{Conclusions and future research} \label{sec6}

Snake-in-the-box codes in $S_n$ under Kendall's $\tau$-metric are useful in the framework of rank modulation for flash memories. In this paper we verify the validity and complete the construction of snake-in-the-box-codes by Horovits and Etzion, with size $M_{2n+1}=\frac{(2n+1)!}{2}-2n+1$. Based on their framework, we further give a construction aiming at a snake-in-the-box-code of size $M_{2n+1}=\frac{(2n+1)!}{2}-2n+3$. We conjecture that our framework is feasible for all odd integers $2n+1\ge7$ and give an example $M_7=2517$. A strict proof for the general validity of our framework is considered for future research.

\end{document}